\documentclass[letterpaper, 10 pt, conference]{ieeeconf}  

\IEEEoverridecommandlockouts                              
\usepackage{booktabs} 
\usepackage{amsfonts}
\usepackage{mathrsfs} 
\usepackage{amsmath} 
\usepackage{psfrag} 
\usepackage{multirow} 
\usepackage{arydshln}
\usepackage{enumerate}
\usepackage{bbm}
\usepackage{graphicx}
\usepackage{color}
\usepackage{mathtools}
\def\mathclap#1{\text{\hbox to 0pt{\hss$\mathsurround=0pt#1$\hss}}}

\newtheorem{assumption}{Assumption}
\newtheorem{proposition}{Proposition}
\newtheorem{definition}{Definition}
\newtheorem{remark}{Remark}
\newtheorem{corollary}{Corollary}
\newtheorem{lemma}{Lemma}
\newtheorem{theorem}{Theorem}

\newcommand{\st}{\text{s.t.}}
\newcommand{\Real}{\mathbb{R}}
\newcommand{\T}{\mathsf{T}}
\newcommand{\co}{\mathsf c}
\newcommand{\rank}{{\rm rank}}
\newcommand{\interior}{{\text{interior}}}
\newcommand{\close}{{\text{closure}}}
\newcommand{\base}{{\bf{e}}}

\newcommand{\clos}{{\text{closure}}}
\newcommand{\inte}{{\text{interior}}}

\newcommand{\alphavec}{\boldsymbol{\alpha}}
\newcommand{\betavec}{\boldsymbol{\beta}}

\newcommand{\Graph}{\mathcal{G}}
\newcommand{\Vertex}{\mathcal{V}}
\newcommand{\VertexG}{\mathcal{V}_{\rm{G}}}
\newcommand{\VertexL}{\mathcal{V}_{\rm{L}}}
\newcommand{\Edge}{\mathcal{E}}
\newcommand{\EdgeI}{\mathcal{E}^{\rm I}}
\newcommand{\EdgeII}{\mathcal{E}^{\rm II}}

\newcommand{\RVertex}{\mathcal{R}}

\newcommand{\nG}{N_{\rm{G}}}
\newcommand{\nL}{N_{\rm{L}}}
\newcommand{\nGL}{N}
\newcommand{\nE}{E}

\newcommand{\CM}{\mathbf{C}}
\newcommand{\B}{\mathbf{B}}

\newcommand{\sgen}{\mathbf{s}^g}
\newcommand{\sload}{\mathbf{s}^l}
\newcommand{\hatsload}{\mathbf{\hat{s}}^l}
\newcommand{\pf}{\mathbf{p}}

\newcommand{\genllim}{\underline{\mathbf{s}}^g}
\newcommand{\genulim}{\overline{\mathbf{s}}^g}

\newcommand{\ang}{\boldsymbol{\theta}}
\newcommand{\f}{{\bf f}}
\newcommand{\M}{{\bf M}}

\newcommand{\pflow}{{\bf p}}
\newcommand{\taueq}{\boldsymbol{\tau}}
\newcommand{\mup}{\boldsymbol{\mu}_+}
\newcommand{\mum}{\boldsymbol{\mu}_-}
\newcommand{\mupm}{\boldsymbol{\mu}}
\newcommand{\lambdap}{\boldsymbol{\lambda}_+}
\newcommand{\lambdam}{\boldsymbol{\lambda}_-}
\newcommand{\lambdapm}{\boldsymbol{\lambda}}
\newcommand{\setf}{\Omega_{\f}}
\newcommand{\setsl}{\Omega_{\sload}}

\newcommand{\setpara}{\Omega}
\newcommand{\para}{{\boldsymbol{\xi}}}
\newcommand{\setslr}{\widetilde{\Omega}_{\sload}}

\newcommand{\setparar}{\widetilde{\Omega}}

\newcommand{\setS}{\mathcal{S}}
\newcommand{\setT}{\mathcal{T}}
\newcommand{\setQ}{\mathcal{Q}}

\newcommand{\OPF}{\mathcal{OPF}}

\newcommand{\JM}{\mathbf{J}}

\newcommand{\setSgen}{\setS_{\rm G}}
\newcommand{\setSbra}{\setS_{\rm B}}

\newcommand{\dpnoise}{Y}
\newcommand{\dpnoises}{{\bf{Y}}}
\newcommand{\query}{\mathcal{M}}
\newcommand{\Prob}{\mathbb{P}}
\newcommand{\Lap}{\mathscr{L}}
\newcommand{\W}{\mathcal{W}}

\newcommand{\Data}{\mathcal{D}}
\newcommand{\data}{{\bf{d}}}

\newcommand{\loosegen}{\aleph}

\title{\LARGE \bf
Differential Privacy of Aggregated DC Optimal Power Flow Data}

\author{Fengyu Zhou, James Anderson, and Steven H. Low
\thanks{This work is funded by NSF grants CCF 1637598, ECCS 1619352, CNS 1545096, ARPA-E through grant DE-AR0000699 and the GRID DATA program, and DTRA through grant HDTRA 1-15-1-0003.}
\thanks{Fengyu Zhou is with the Department of Electrical Engineering, California Institute of Technology, Pasadena, CA, 91125. Email: {\tt\small f.zhou@caltech.edu}}%
\thanks{James Anderson is with the Department of Computing and Mathematical Sciences, California Institute of Technology, Pasadena, CA, 91125. Email: {\tt\small james@caltech.edu}}%
\thanks{Steven H. Low is with Department of Electrical Engineering and the Department of Computing and Mathematical Sciences, California Institute of Technology, Pasadena, CA, 91125. Email: {\tt\small slow@caltech.edu}}%
}

\begin{document}

\maketitle
\thispagestyle{empty}
\pagestyle{empty}

\begin{abstract}
We consider the problem of privately releasing  aggregated network statistics obtained from solving a DC optimal power flow (OPF) problem. It is shown that the mechanism that determines the noise distribution parameters are linked to the topology of the power system and the monotonicity of the network. We derive a measure of ``almost'' monotonicity and show how it can be used in conjunction with a linear program in order to release aggregated OPF data using the differential privacy framework.
\end{abstract}

\section{Introduction}
Realistic and publicly available power network models based on real data 
are important for the research community.  
One of the difficulties in developing such a model is that grid operators are
reluctant to disclose consumer data or any information that may be commercially sensitive. 
%
Differential privacy, first developed in \cite{Dwo2006calibrating,Dwo2008differential,Dwo2014algorithmic}, has been widely used to evaluate the privacy loss for individual users in a dataset. It has recently been used by the researchers in the  power systems community for use in applications such as distributed algorithms for EV charging \cite{Han2014differentially}, power system data release \cite{FioV2018constrained}, and load management \cite{Hal2017architecture}. 

In our work, we consider the differential privacy of power systems induced by an Optimal Power Flow (OPF) problem. In this context, the optimal generation can be viewed as a function of the loads. Typically generation data is publicly available. In contrast, load data can reveal consumer habits and other commercially sensitive information, and thus we aim to keep it private. We aim to prevent changes in generation data from disclosing sensitive load data.
Instead of proposing new mechanisms, for a given network we study how much noise is required to be added to the data in order to achieve a certain level of differential privacy for existing mechanisms such as the Laplace mechanism. We introduce the concept of $(\delta, \varepsilon)$-monotonicity, a metric that is central  to our differential privacy analysis. We also show how it is affected under different system topologies. Finally we present examples of three systems with different topologies and thus different monotonic characterizations, i.e., different $(\delta,\varepsilon)$ parameters. For each system we show that to preserve the same level of differential privacy, the required amount of noise implied by our theorem is very different for each example. We hope that such theoretical guarantees will not only guide the design of differentially private power systems, but also encourage greater data sharing and cooperation between grid operators and academia in the future.

We stress that the aim of this work is not to show that a linear program can be made differentially private. There are numerous results in this area, see for example~\cite{HsuRRU14, HsuHRW16, WaiJD12}. In the setting we consider, the grid operator will solve an appropriate optimization problem and will have access to \emph{all the data}. The results we provide will be based on using the Laplace mechanism to release this data privately. We note that there are other mechanisms available (e.g. the exponential and Gaussian mechanisms, as well as some that allow one to specify the support of a distribution) and indeed some may be better suited for this particular application. However, the Laplace mechanism is used in this paper as it most clearly links the key concepts of monotonicity, sensitivity, and topology and their relationship to privacy - this dependence has until now not been identified.
\section{Background}

\subsection*{Notation}
Vectors and matrices are typically written in bold while scalars are not. Given two vectors $\mathbf{a},\mathbf{b} \in \Real^n$, $\mathbf a\ge \mathbf b$ denotes the element-wise partial order $\mathbf{a}_i \ge \mathbf{b}_i$ for $i=1,\hdots,n$. For a scalar $k$, we define the projection operator $[k]^-:= \min \{0,k\}$. We define $\|\mathbf x\|_0$ as the number of non-zero elements of the vector $\mathbf x$. 
For $\mathbf X \in \Real^{n\times m}$, the restriction $\mathbf{X}_{\{1,3,5\}}$ denotes the $3\times m$  matrix composed of stacking rows $1,3$, and $5$ on top of each other. We will frequently use a set to describe the rows we wish to form the restriction from, in which case we assume the elements of the set are arranged in increasing order.
We will use $\base_{\scriptscriptstyle m}$ to denote the $m^{\text{th}}$ standard basis vector,  its dimension will be clear from the context. Finally, let $[m]:=\{1,2,\dots,m\}$ and $[n,m]:=\{n,n+1,\dots,m\}$.

\subsection{System Model}
Consider a power network modeled by an undirected graph $\Graph(\Vertex, \Edge)$, where $\Vertex:=\VertexG\cup\VertexL$ denotes the set of buses which can be further classified into generators in set $\VertexG$ and loads in set $\VertexL$, and $\Edge\subseteq \Vertex\times\Vertex$ is the set of all branches linking those buses. We will later use the terms (graph, vertex, edge) and (power network, bus, branch) interchangeably. Suppose $\VertexG\cap\VertexL=\emptyset$ and there are $|\VertexG|=:\nG$ generator and $|\VertexL|=:\nL$ loads, respectively.
For simplicity, let $\VertexG=[\nG]$, $\VertexL=[\nG+1,\nG+\nL]$.  Let $\nGL=\nG+\nL$. Without loss of generality, $\Graph$ is a connected graph with $|\Edge|=:\nE$ edges labelled as $1,2,\dots,\nE$.
Let $\CM\in\Real^{\nGL\times\nE}$ be the signed incidence matrix. 
Let $\B={\rm diag}(b_1,b_2,\dots,b_E)$, where $b_e>0$ is the susceptance for branch $e$. As we adopt a
DC power flow model, all branches are assumed lossless. Further, we denote the generation and load as $\sgen\in\Real^{\nG}$, $\sload\in\Real^{\nL}$, respectively.
Thus $\sgen_i$ refers to the generation on bus $i$ while $\sload_i$ refers to the load on bus $\nG+i$. We will refer to bus $\nG+i$ simply as load $i$ for simplicity. 
The power flow on branch $e\in\Edge$ is denoted as $\pf_e$, and $\pf:=[\pf_1,\dots,\pf_{\nE}]^{\T}\in\Real^\nE$ is the vector of all branch power flows.
The following assumption is made to simplify the analysis.
\begin{assumption}\label{A1}
There are no buses in the network that are both loads and generators. Formally,
$\VertexG\cap\VertexL=\emptyset$.
\end{assumption}

The above assumption is not restrictive under the lossless assumption in DC power flow.
We can always split a bus with both a generator and a load into a bus with
only the generator connected to another bus with only the load, and
connect all the neighbors of the original bus to that load bus.

\subsection{Optimal Power Flow}
We focus on the DC OPF problem with a linear cost function \cite{Woo2012power}. That is to say, the voltage magnitudes are assumed to be fixed and known. Without loss of generality, we assume all the voltage magnitudes to be $1$. The decision variables are the voltage angles denoted by vector $\ang\in\Real^{\nGL}$ and power generations $\sgen$, given loads $\sload$. The DC OPF takes the following form:
\begin{subequations}
\begin{eqnarray}
\underset{\sgen, \ang}{\text{minimize}}  && \f^{\T}\sgen
\label{eq:opf1.a}\\
\text{ subject to }& & \ang_1 = 0
\label{eq:opf1.b}\\
& &  \CM\B\CM^{\T} \ang = 
\left[ \begin{array}{c}
\sgen \\
-\sload
\end{array} \right] 
\label{eq:opf1.c}\\
& & \genllim \leq\sgen\leq \genulim
\label{eq:opf1.d}\\
& &  \underline{\pflow}\leq\B\CM^{\T}\ang\leq\overline{\pflow}.
\label{eq:opf1.e}
\end{eqnarray}
\label{eq:opf1}
\end{subequations}
Here, each entry of $\f\in\Real^{\nG}$ is the unit cost for a generator, and bus $1$ is selected as the slack bus with fixed voltage angle $0$. In \eqref{eq:opf1.c}, we let the injections for generators be positive while the injections for loads be the negation of $\sload$.
The upper and lower limits for the generation are set as $\genulim$ and $\genllim$, respectively, and $\overline{\pflow}$ and $\underline{\pflow}$ are the limits for branch power flow. 
We assume that~\eqref{eq:opf1} is well posed, i.e. $\genulim>\genllim$, $\overline{\pflow}>\underline{\pflow}$.

\subsection{Differential Privacy}
In this subsection, we introduce the concept of differential privacy as a method for evaluating the privacy status of a dataset. 
In general, a differentially private dataset can protect the privacy of each individual user by adding noise to  database queries such that the change in a single record cannot be effectively detected \cite{Dwo2006calibrating,Dwo2008differential,Dwo2014algorithmic}.
Suppose $\Data^n$ is the data space for $n$ users. Then a data element is $\data\in\Data^n$. A \emph{query} is a function $\tilde{\query}:\Data^n \rightarrow  \Real^r$. Examples include ``count" functions, e.g. return the number of records in the database where property $y$ holds ($r=1$). Other examples include statistical queries such as computing mean and variance. A mechanism  ${\query}:\Data^n \rightarrow  \Real^r$
is a randomized function of $\data$ which releases the result of the query combined  with an appropriately defined level of noise. For example,  a mechanism $\query$ can return the value
$\query(\data) := \tilde{\query}(\data) + \dpnoises$ for an appropriately chosen noise $\dpnoises$.


\begin{definition}[\cite{Dwo2006calibrating}]\label{Df:DP}
The mechanism $\query$ is said to preserve  $\varrho$-differential privacy if and only if $\forall \data',\data''\in\Data^n$ such that $\|\data'-\data''\|_0\leq1$, 
and $\forall \W\subseteq\Real^r$, we have
\begin{align*}
\Prob\{\query(\data')\in\W\}
\leq \exp(\varrho)\cdot\Prob\{\query(\data'')\in\W\}.
\end{align*}
\end{definition}

A mechanism $\query$
that satisfies the properties of Definition~\ref{Df:DP} ensures that the addition or removal of a single entry to the database does not change (much) the outcome of the query.

The Laplace mechanism is a popular choice relying on the symmetric Laplace distribution $\Lap(\cdot)$. For a random variable $X\sim \Lap(b)$ the probability density function is given by
\begin{equation*}
f_X(x|b) = \frac{1}{2b}\exp \left(\frac{-|x|}{b} \right),
\end{equation*} 
and $X$ has variance $\sigma^2 = 2b^2$. Intuitively, as $b$ increases, the distribution flattens and spreads symmetrically about the origin. The Laplace mechanism is defined
by  $\query(\data) := \tilde{\query}(\data)+\dpnoises$ where $Y_i\sim \Lap(\Delta_1/\varrho)$ are independent and identically distributed for $i=1,\hdots, r$ and $\Delta_1$ is the $L_1$-sensitivity of the query $\tilde{\query}$:
\begin{equation}\label{eq:L1sens}
\Delta_1 = \underset{\|\data'-\data''\|_0\leq1}{\text{maximize}} \quad \|\tilde{\query}(\data')-\tilde{\query}(\data'')\|_1.
\end{equation}
The following theorem explains the importance of the Laplace mechanism
~\cite{Dwo2006calibrating}:
\begin{theorem}\label{thm:lapDP}
For $\tilde{\query}:\Data ^n \rightarrow \Real^r$, the Laplace mechanism defined by
$\Lap({\Delta_1/\varrho})$ provides $\varrho$-differential privacy.
\end{theorem}
From the theorem and the definition of the Laplace distribution, it can be seen that for a fixed privacy level (specified by $\varrho$), as the sensitivity increases, the mechanism responds by adding noise drawn from a distribution of increasing variance. Fortunately, many queries of interest have low sensitivity;
e.g., counting queries and sum-separable functions have $\Delta_1=1$, 
\section{Preliminaries}
\subsection{OPF Operator}\label{sec:assump}

We now fix the topology and susceptances of the power network.
Let $\para:=[(\genulim)^{\T},(\genllim)^{\T},\overline{\pflow}^{\T},\underline{\pflow}^{\T}]^{\T}\in\Real^{2\nG+2\nE}$ be the vector of system limits.
Define
\begin{align*}
\setpara:=\{\para | \genllim\geq 0, \text{\eqref{eq:opf1.b}-\eqref{eq:opf1.e} are feasible for some }\sload>0\}.
\end{align*} 
{For each $\para\in\setpara$, define
\begin{align*}
\setsl(\para)&:=\{\sload | \sload>0, \text{\eqref{eq:opf1.b}-\eqref{eq:opf1.e} are feasible}\},\\
\setslr(\para)&:=\{\sload\!\!\in\!\setsl(\para)|\eqref{eq:opf1}~\text{has}~\nG\!-\!1
~\text{binding inequalities}\}.
\end{align*} 
Here $\setsl$ is convex and nonempty.
When we fix $\para$ and there is no confusion, we use $\setsl$ and $\setslr$ instead.

We now define the operator $\OPF$, which will be used throughout the rest of the paper.
\begin{definition}
Let the set valued operator $\OPF:\setsl\rightarrow 2^{\Real^{\nG}}$ be the mapping such that $\OPF(\bf{x})$ is the set of optimal solutions to \eqref{eq:opf1} with parameter $\sload=\bf{x}$.
\footnote{Here, $2^{\Real^\nG}$ indicates the power set of $\Real^\nG$.}
\end{definition}
We adopt the following assumption to simplify $\OPF$.
Fix $\B, \CM$ and $\para$, let $\setf$ be the set of $\f$ such that $\forall \sload\in\setsl$,
\begin{itemize}
\item \eqref{eq:opf1} has a unique solution;
\item the Lagrange multipliers of the KKT conditions (Appendix \ref{app:Assumption1}, eq.\eqref{eq:KKT}) satisfy 
\begin{align}\label{eq:multiplier1}
\|\mup\|_0+\|\mum\|_0+\|\lambdap\|_0+\|\lambdam\|_0\geq \nG-1.
\end{align}
\end{itemize}
\begin{assumption}\label{A:vectorf}
The objective vector $\f$ is in $\setf$, i.e, $\f$  always guarantees the uniqueness of the solution to \eqref{eq:opf1} for all $\sload\in\setsl$.
\end{assumption}
The motivation  for Assumption \ref{A:vectorf} is technical and deferred to the Appendix.
\begin{remark}
Under Assumption \ref{A:vectorf}, the value of $\OPF$ is always a singleton, 
so we can consider $\OPF({\bf x})$ as a function mapping ${\bf x}$ to the unique optimal solution of \eqref{eq:opf1} with parameter $\sload=\bf{x}$.
Since the solution set to the parametric linear program is both upper and lower hemi-continuous \cite{Zha1990note}, $\OPF$ is continuous.
\end{remark}
\begin{remark}
Intuitively, $\setpara$ and $\setsl$ contain the parameters that make \eqref{eq:opf1} feasible,
while $\setslr$ and $\setf$ also provide $\OPF$ with good properties such as uniqueness and differentiability.
\end{remark}

\subsection{System Monotonicity}
System monotonicity characterizes how the optimal 
generation reacts to a change in load.  It sheds lights on the {$L_1$-sensitivity}.  
\begin{definition}\label{Def:monotonicity}
A power system is said to be \it{monotone} if 
$\forall\alphavec,\betavec\in\setsl$ such that $\alphavec\geq\betavec$  and $\|\alphavec-\betavec\|_0=1$,
we have $\OPF(\alphavec)\geq\OPF(\betavec)$.
\end{definition}
In the DC power flow model, $\sum_i\OPF_i(\alphavec)=\sum_j\alphavec_j\geq\sum_j\betavec_j=\sum_i\OPF_i(\betavec)$,
i.e., the total generation to meet demand $\alphavec$ is greater than or equal to the total generation to meet demand $\betavec$, but the equalities in Definition \ref{Def:monotonicity} are stronger. 
They are element-wise, i.e., a system is monotone if \emph{all} generations will 
 increase or remain unchanged when any single load increases.
 This is often too stringent a requirement.  We are interested in approximately
 monotone systems, formalized in the following definition.
\begin{definition}\label{Def:monotonicity2}
For $\delta>0,\varepsilon\geq 0$, a power system is said to be \it{$(\delta,\varepsilon)$-monotone} if 
$\forall\alphavec,\betavec\in\setsl$ such that $\betavec+\delta\cdot\mathbf{1}\geq\alphavec\geq\betavec$ and $\|\alphavec-\betavec\|_0=1$,
we have $\sum_{i=1}^{\nG}[\OPF_i(\alphavec)-\OPF_i(\betavec)]^-\geq-\varepsilon$.
We refer to $(\delta,\varepsilon)$ as a monotonicity pair.
\end{definition}

By definition, a monotone system  is always $(\delta,0)$-monotone for any positive $\delta$. In the next subsections, we will study the $\OPF$ derivative and then relate it to monotonicity.

\begin{figure}
\centering
\includegraphics[width=0.8\columnwidth]{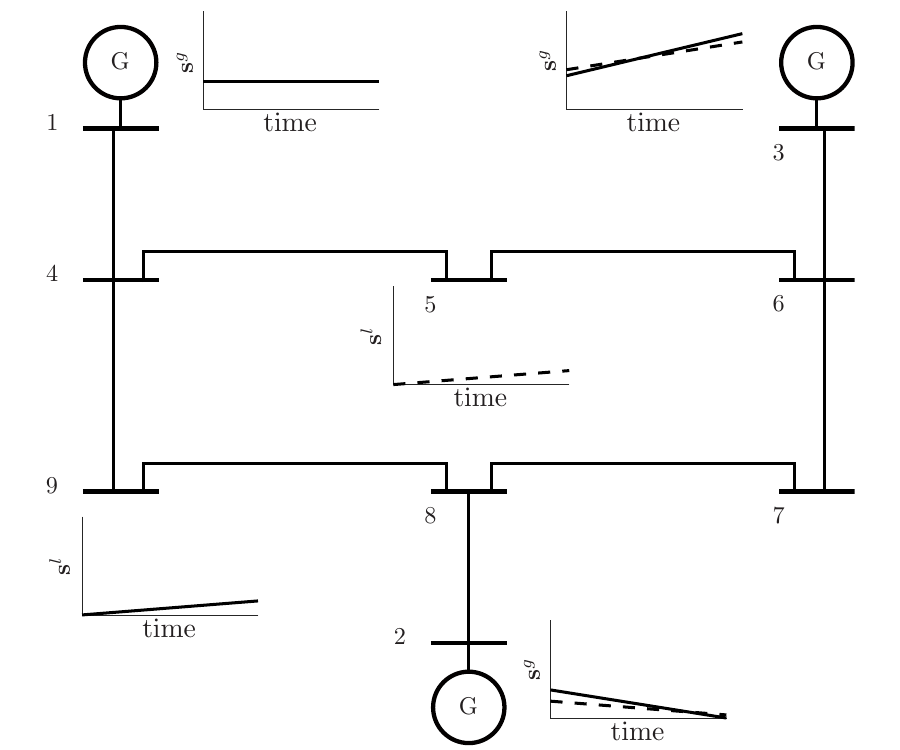}
\caption{IEEE 9-bus case. Dashed and solid curves show how the optimal generations 
change as loads on bus 5 and bus 9 increase. Bus 1 has constant generation 
since its generation has reached its upper limit.}
\vspace{-1em}
\label{fig:mono_case9}
\end{figure}

Here we use the IEEE 9-bus testcase as an example to illustrate the concept of monotonicity.
As shown in Figure \ref{fig:mono_case9}, increasing the load on either bus 5 (dashed curves) or bus 9 (solid curves) 
will lead to production decrease in generator 2. 
Thereby, IEEE 9-bus testcase is not monotone.
A more careful analysis shows that for any $\delta>0$, the system is actually $(\delta,2.01\delta)$-monotone, meaning the total decrease in the optimal generation will not exceed 2.01 times the increase in the load.

\subsection{Determining Monotonicity}
Monotonicity as in Definition~\ref{Def:monotonicity} does not hold for general networks. 
In this subsection we characterize topologies that are monotone.
In particular, we show that radial networks are monotone.

An equivalent definition of monotonicity is that the derivative\footnote{We adopt the following notation: $\partial_{\sload}\OPF(\sload) :=\frac{\partial \OPF (\sload)}{\partial \sload} = \frac{\partial \sgen}{\partial \sload}$, which has dimension $\nG \times \nL$.} 
$\partial_{\sload}\OPF(\sload)$ of the corresponding $\OPF$ operator is element-wise nonnegative (when it exists).
Let $\setSgen(\sload)$ and $\setSbra(\sload)$ denote the set of generators and branches that are 
binding, respectively, for a given $\sload$, i.e.
\begin{align*}
&\setSgen(\sload):=\{i\in\VertexG:\sgen_i\in\{\genllim_i,\genulim_i\},\\
&\setSbra(\sload):=\{e\in\Edge:\pf_e\in\{\underline{\pflow}_e,\overline{\pflow}_e\}\}.
\end{align*}
When there is no danger of confusion, we will write $\setSgen$ and $\setSbra$ for simplicity.
\begin{assumption}\label{A:derivative}
The set $\setslr$ is dense in $\setsl$.
For $\sload\in\setslr$, the derivative
$\partial_{\sload}\OPF(\sload)$ always exists, and the sets $\setSgen$ and $\setSbra$ do not change in a neighborhood of $\sload$.
\end{assumption}

We show in Appendix \ref{app:Assumption_der} that Assumption \ref{A:derivative} is mild.

Returning to the graph $\Graph(\Vertex,\Edge)$, we divide $\Edge$ into two disjoint sets:
\begin{align}
\nonumber
&\EdgeI:=\{e\in\Edge~|~\Graph(\Vertex,\Edge\backslash \{e\})~\text{is not connected}\}\\
\nonumber
&\EdgeII:=\Edge\backslash \EdgeI.
\end{align} 
Links in $\EdgeI$ are called bridges in $\Graph$. 
In general, it is possible that $\EdgeI = \emptyset$,
e.g., when $\Graph$ is a ring. 
The next result connects monotonicity to network topology.
\begin{theorem}\label{Thm:treeTopl}
For any $\sload\in\setslr$ such that $\setSbra(\sload)\subseteq\EdgeI$, we have $\partial_{\sload}\OPF(\sload)\geq 0$, i.e., the system is monotone.
\end{theorem}
\begin{proof}
See Appendix \ref{app:proof_of_Thm_4}.
\end{proof}

Thereom \ref{Thm:treeTopl} directly implies  the following corollaries.

\begin{corollary}\label{Co:treeMono}
Power networks whose graphs $\Graph$ are trees are monotone.
\end{corollary}
\begin{corollary}
If all the possible branch flow bottlenecks\footnote{We define a bottleneck to be any edge $e\in\Edge$ such that $\exists \sload\in\setslr$ where the optimal power flow $\pf_e \in \{{\underline{\pflow}}_e ,{\overline{\pflow}}_e\}$.} in the power system are in $\EdgeI$, then the system is monotone.
\end{corollary}

In general, when the cycles in the graph are not adjacent to each other, the monotonicity pair $(\delta,\varepsilon)$ can be efficiently estimated. The algorithm and its proof will be presented in the journal version of this paper. 


\section{OPF Privacy}
\subsection{Motivation and Definition}
Ideally both the generations $\sgen$ and loads $\sload$ are  available 
for the research community to build realistic power system models from.
However, load data may contain sensitive information, and hence it is desirable 
to preserve the privacy of $\sload$.

Suppose $\query(\sgen,\sload)$ is a (randomized) function of $(\sgen,\sload)$, and acts as the mechanism of the data. 
It is reasonable to assume that $\sgen$ is always chosen as the
unique optimal solution to the OPF problem, i.e., $\sgen=\OPF(\sload)$.
Then we can write $\query(\sgen,\sload)$ as $\query(\OPF(\sload),\sload)$.
For simplicity, we denote it as $\query(\sload)$.
The privacy problem is to design a mechanism that hides individual load 
changes
when the database containing the vectors $\sgen, \sload$ is queried. 
We let  $\Delta$ denote the changes to an individual load, 
i.e.,  $\sload_i \leftarrow \sload_i \pm\Delta$ for some  $\Delta>0$.
To address this problem, we introduce a modified version of differential privacy: 
\begin{definition}
For $\Delta,\varrho>0$, the mechanism $\query$ preserves $(\Delta,\varrho)$-differential privacy 
\footnote{The definition of $(\Delta,\varrho)$-differential privacy in this paper is different from the standard definition used in \cite{Dwo2014algorithmic}. In particular, the second parameter does not refer to an additive term in Definition~\ref{Df:DP}, but rather a bound on the $\ell_1$-sensitivity of the loads.}
if and only if $\forall (\sload)',(\sload)''$ such that $\|(\sload)'-(\sload)''\|_0\leq1$ and $\|(\sload)'-(\sload)''\|_1\leq\Delta$, 
and $\forall \W\subseteq\Real^r$, we have
\begin{align*}
\nonumber
\Prob\{\query((\sload)')\in\W\}
\leq \exp(\varrho)\cdot\Prob\{\query((\sload)'')\in\W\}.
\end{align*}
\end{definition}

Theorem~\ref{thm:lapDP} can be readily extended to our $(\Delta,\varrho)$-differential privacy.
\begin{lemma}\label{lm:dp}
Let $\tilde{\query}(\sgen,\sload)$ be a deterministic query. The mechanism $\query=\tilde{\query}+\dpnoises$, with $Y_i$ drawn i.i.d.  from $\Lap(\Delta_1/\varrho)$, preserves $(\Delta,\varrho)$-differential privacy if for any $(\sload)',(\sload)''$ such that $\|(\sload)'-(\sload)''\|_0\leq1$ and $\|(\sload)'-(\sload)''\|_1\leq\Delta$,  $\Delta_1$ satisfies
$ \|\tilde{\query}((\sload)')-\tilde{\query}((\sload)'')\|_1\leq \Delta_1.$
\end{lemma}
\subsection{Queries for Power Systems}
We investigated a few commonly used statistics for power systems provided by U.S. Energy Information Administration (EIA) \cite{EIAdata} and  
French transmission system operator (RTE) \cite{RTEdata}.
Here we list a few of them and view them as the potential queries for power system data.
\begin{itemize}
\item Regional aggregated generation and load: total generation or load within a region regulated by each grid operator.
\item Power generation by energy source: total generation provided by each individual source of energy such as solar or wind.
\item Inter-regional flows: power traded among different regions.
\end{itemize}

Most of those statistics can be regarded as some linear functions of the generation $\sgen$ and load $\sload$.
In the next subsection, we will focus on the example of an aggregation query, which is a generalized model for many statistics listed above.

\subsection{Aggregation Query}
In \cite{And2018disaggregation}, we propose a method to release 
load and generation data in a way that attempts to strike a balance between 
the privacy of data owners and the need of the research community for realistic 
samples.  The method consists of two steps.  
First, instead of $\sgen$ and $\sload$, the data owner releases their
aggregations over discrete regions of the network.
Second, a disaggregation algorithm is used to estimate 
the loads and generations based on the released aggregated data.
 In this section, we study how differential privacy is preserved for the aggregation query.
 See \cite{FioV2018constrained} for another approach.

Suppose the buses in $\Vertex$ are partitioned into $r$ regions 
$\RVertex_1,\RVertex_2 \dots ,\RVertex_r$, where $\RVertex_i\subseteq \Vertex$ 
is the set of bus IDs in region $i$.
Let the aggregation query for region $i$ be
\begin{align*}
\tilde{\query}_i^g = \sum\limits_{j\in \RVertex_i}{\sgen}_j,\quad
\tilde{\query}_i^l = \sum\limits_{j+\nG\in \RVertex_i}{\sload}_j,
\end{align*}
The system operator discloses a noisy version of the aggregation query, denoted as $\query_i^g=\tilde{\query}_i^g+Y_i^g$ and $\query_i^l=\tilde{\query}_i^l+Y_i^l$. 
Here, $\dpnoise_i^g$ and $\dpnoise_i^l$ are 
 independent random variables and are intentionally added to ensure privacy.
Let 
\begin{align}\label{eq:query}
\query(\sgen,\sload)=[\query_1^g,\dots,\query_r^g,\query_1^l,\dots,\query_r^l]^{\T}
\end{align}
be the Laplace mechanism for this aggregation query.
Since the support of Laplace distribution is unbounded, there is a chance that the mechanism will change the signs of the query and make the output data unrealistic.
In practice, one can easily use the exponential mechanism to solve this issue by defining a quality function which penalizes the data with wrong signs \cite{McS2007mechanism}.
In this paper we will not provide the details as space is limited and our primary motivation is to show how system monotonicity, sensitivity, and topology are related to the data privacy via the Laplace mechanism.
We will see in Section \ref{sec:Simulation} that networks that are likely to encounter sign errors tend to be far from monotone, 
in which case it is hard to preserve both the privacy and data quality no matter which mechanism is applied due to high sensitivity of the system.

Lemma \ref{lm:dp} and Definitions \ref{Def:monotonicity} and \ref{Def:monotonicity2} immediately imply the following two properties of  \eqref{eq:query}.
\begin{theorem}\label{thm:mono2}
Suppose the system is $(\Delta,\varepsilon)$-monotone.  
The mechanism \eqref{eq:query}
where $Y_i^g$ and $Y_i^l$ are drawn i.i.d. from $\Lap(2(\Delta+\varepsilon)/\varrho)$
preserves $(\Delta,\varrho)$-differential privacy.
\end{theorem}
\begin{proof}
By Definition \ref{Def:monotonicity2}, we have 
\begin{align*}
\|(\sgen)'-(\sgen)''\|_1\leq \|(\sload)'-(\sload)''\|_1+2\varepsilon\leq \Delta+2\varepsilon.
\end{align*}
Thus,
\begin{align*}
\nonumber
&\|\tilde{\query}((\sload)')-\tilde{\query}((\sload)'')\|_1\\
\leq&\|(\sgen)'-(\sgen)''\|_1+ \|(\sload)'-(\sload)''\|_1\leq 2\Delta+2\varepsilon.
\end{align*}
Then the conclusion is implied by Lemma \ref{lm:dp}.
\end{proof}

\begin{corollary}\label{thm:mono1}
Suppose the power system is monotone.  The mechanism \eqref{eq:query}
where $Y_i^g$ and $Y_i^l$ are drawn i.i.d. from $\Lap(2\Delta/\varrho)$
preserves $(\Delta,\varrho)$-differential privacy.
\end{corollary}

\begin{remark}
The monotonicity pairs for a fixed system are not unique. In Theorem \ref{thm:mono2}, for any given $\Delta>0$, there always exists $\varepsilon>0$ such that the system is $(\Delta,\varepsilon)$-monotone.
\end{remark}

\begin{remark}
For the Laplace distributions given in Theorems \ref{thm:mono1} and Corollary \ref{thm:mono2}, the level of differential privacy
is independent of how the aggregation regions are divided and how the data are aggregated.
In particular, the amount of noise required relies on neither the number $r$ of regions 
nor the number of buses in each region.
\end{remark}

The following example shows why we want the level of differential privacy to be independent of the region division.
Consider a trivial mechanism which can preserve the same $(\Delta,\varrho)$-differential privacy by adding i.i.d. Laplace noise drawn from $\Lap(\Delta/\varrho)$ to each individual load
and then solving an OPF problem with the  noisy load data to obtain the  generations.
This mechanism can guarantee  $(\Delta,\varrho)$-differential privacy, assuming that the noisy load makes $\OPF$ feasible and yields a unique solution.
Then, from the central limit theorem, the equivalent noise added to $\tilde{\query}_i^l$ would converge in probability to the  Gaussian distribution $\mathcal{N}(0,2(\Delta/\varrho)^2|\RVertex_i\cap\VertexL|)$.
The variance of this distribution depends on the size of the region and can grow rapidly if the region is large, in comparison to the (equivalent) Laplacian distribution $\Lap(2(\Delta+\varepsilon)/\varrho)$ given in Theorem \ref{thm:mono2}.
As for the equivalent noise added to $\tilde{\query}_i^g$, we can give a rough estimation. 
Since the noise added to each load is on the order of $\Delta/\varrho$, the noise vector added to the load vector has the $L_1$-norm on the order of $\Delta/\varrho|\RVertex_i\cap\VertexL|$.
Assume that $(\Delta,\varepsilon)$-monotone system can potentially amplify the noise in the load vector by a factor of roughly $1+2\varepsilon/\Delta$,
the equivalent noise added to $\tilde{\query}_i^g$ could be on the order of $(\Delta+2\varepsilon)/\varrho|\RVertex_i\cap\VertexL|$,
which also depends on the size of the region and can potentially be quite large.

\subsection{Generalization}

In general, for an arbitrary query not necessarily the aggregation query, 
the $L_1$-sensitivity $\Delta_1$ in \eqref{eq:L1sens} depends on the properties of both $\tilde{\query}$ and $\OPF$. 
When $\tilde{\query}$ is the aggregation query, the problem boils down to the monotonicity of $\OPF$, as shown in the previous subsection.
However, for general $\tilde{\query}$, the estimation of $\Delta_1$ may require a more 
careful analysis of the structure of $\tilde{\query}$.
The next result provides a rough estimation of the required amount of noise for
differential privacy.  Its proof is omitted due to space limitation.
\begin{theorem}\label{thm:generalcase}
Suppose a power system is $(\Delta,\varepsilon)$-monotone, and all elements of the Jacobian
matrix $\JM_{\tilde{\query}}$ with respective to $\tilde{\query}\in\Real^{r}$
are upper bounded by the same constant $U$.
Then the mechanism
$\query=\tilde{\query}+\dpnoises$, where all $Y_i$ are drawn from independent Laplace distribution
$\Lap(2Ur{(\Delta+\varepsilon)}/{\varrho})$,
 preserves $(\Delta,\varrho)-$differential privacy.
\end{theorem}

In the aggregation case, $U=1$ and $r$ is the number of regions. 
Comparing Theorem \ref{thm:mono2} and Theorem \ref{thm:generalcase}, 
the required Laplace noise is reduced by a factor of $r$ in Theorem \ref{thm:mono2}
which exploits the simple structure of the aggregation function.
\section{Simulation}\label{sec:Simulation}
\begin{figure}
\centering
\includegraphics[width=\columnwidth]{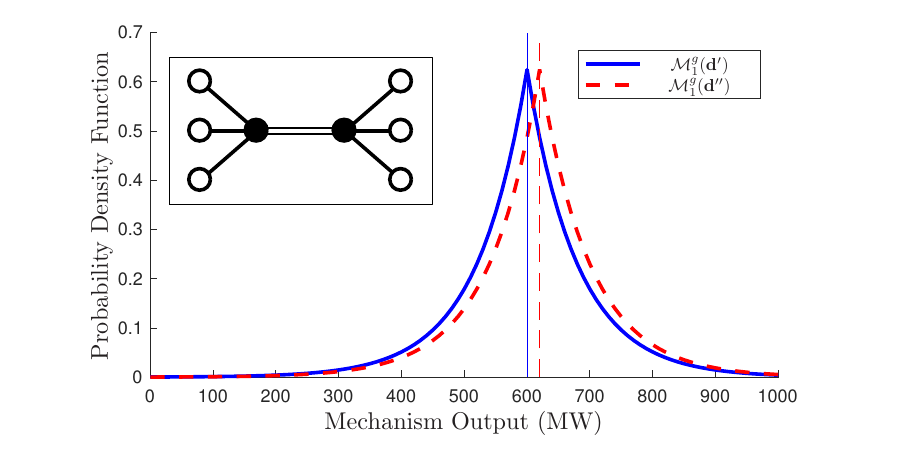}
\vspace{-2em}
\caption{The embedded diagram shows the topology of a radial power network, where black and white nodes indicate generators and loads, respectively. The double-line edge is the bottleneck of the system, and in our example, its line flow constraint is always binding. 
The vertical lines indicate the ground-truths of the queries for two datasets 
whose difference we want to hide.  The curves are the probability density functions of 
the mechanism outputs that contain Laplace noise.}
\vspace{-1em}
\label{fig:DP_tree}
\end{figure}
\subsection{Radial Network}
First, we apply the mechanism \eqref{eq:query} 
to a radial power network (embedded image in Figure  \ref{fig:DP_tree}), 
i.e., network with a tree topology.
Corollary \ref{Co:treeMono} implies that the system is monotone, and by Theorem \ref{thm:mono1}, the noise should be drawn independently from the Laplace distribution $\Lap(2\Delta/\varrho)$ so as to preserve 
 $(\Delta,\varrho)$-differential privacy.  In this simulation, we set $\Delta=20$ (MW) and $\varrho=0.5$. The interpretation is that  any two datasets whose difference we would like to hide 
 should differ in any one load by at most $20$ MW.  
 This example has been constructed so that the double-line edge in Fig. \ref{fig:DP_tree} is a bottleneck (i.e., a binding constraint in the solution of \eqref{eq:opf1}).
 According to Appendix \ref{app:proof_of_Thm_4}, this bottleneck splits the tree into two subtrees and each subtree contains exactly one generator which is not saturated. Provided the OPF problem remains feasible, any change in the load will directly lead to the same amount of change in the generator which resides in the same subtree as the changing load.

Specifically, if the load on the left increases by $\Delta$, the left generator will increase its
generation by $\Delta$ while the right generator will remain unchanged. 
Hence the ground-truths of the aggregation queries (shown by vertical lines in Fig. \ref{fig:DP_tree}) 
are separated by $20$ MW, the same as $\Delta$. The density
functions shown in Fig. \ref{fig:DP_tree} are sharper than those of networks with cycles, as shown in Figures \ref{fig:DP_case9} and \ref{fig:DP_cycle}.

\begin{figure}
\centering
\includegraphics[width=\columnwidth]{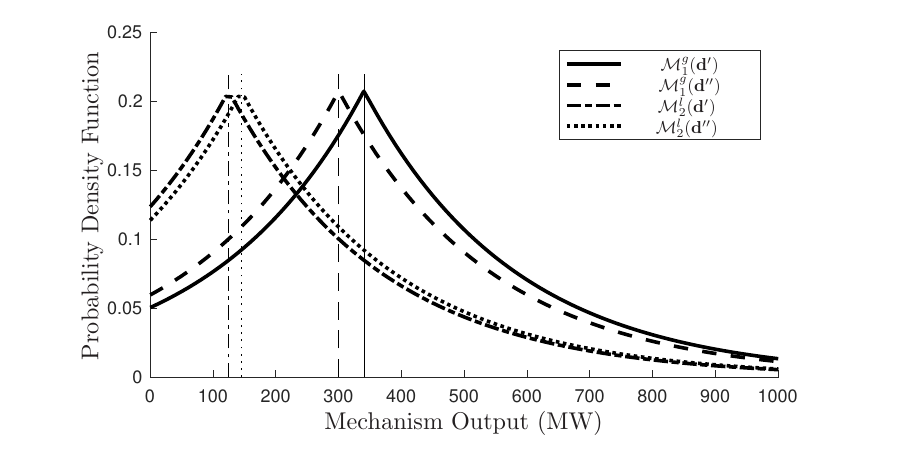}
\vspace{-2em}
\caption{Differential privacy for IEEE 9-bus case. The vertical lines indicate the ground-truths, and the curves show the probability distribution of the mechanism outputs. Only the aggregated generation in region 1 and the aggregated load in region 2 are presented in the figure.}
\label{fig:DP_case9}
\end{figure}
\subsection{IEEE 9-Bus Network}
As we mentioned, the IEEE 9-bus network is $(\Delta,2.01\Delta)$-monotone for any positive $\Delta$.
We again set $\Delta=20$ (MW), $\varrho=0.5$ and divide the system into two regions.
In our simulations, region 1 contains buses 1, 2, 4, 5, 6, while region 2 contains buses
3, 7, 8, 9. Figure \ref{fig:DP_case9} shows the probability density functions
for the aggregation mechanism when the load on bus 9 increases by $\Delta$. 
The difference between $\query_2^l(\data')$ and $\query_2^l(\data'')$ comes directly from the change on bus 9, but the difference between $\query_1^g(\data')$ and $\query_1^g(\data'')$ is mainly due to the fact that generator 3 has to increase its generation so as to compensate for the  decrease in generation on bus 2. 
Hence the distributions for $\query_1^g(\data')$ and $\query_1^g(\data'')$ are further apart compared to the distributions for $\query_2^l(\data')$ and $\query_2^l(\data'')$. As a result, to preserve the same level of differential privacy, the required noise magnitude is greater than what would have
been needed if the system were monotone. The distributions in Figure~\ref{fig:DP_case9} are 
indeed flatter than those in Figure \ref{fig:DP_tree}.

\begin{figure}
\centering
\vspace{-0.5em}
\includegraphics[width=\columnwidth]{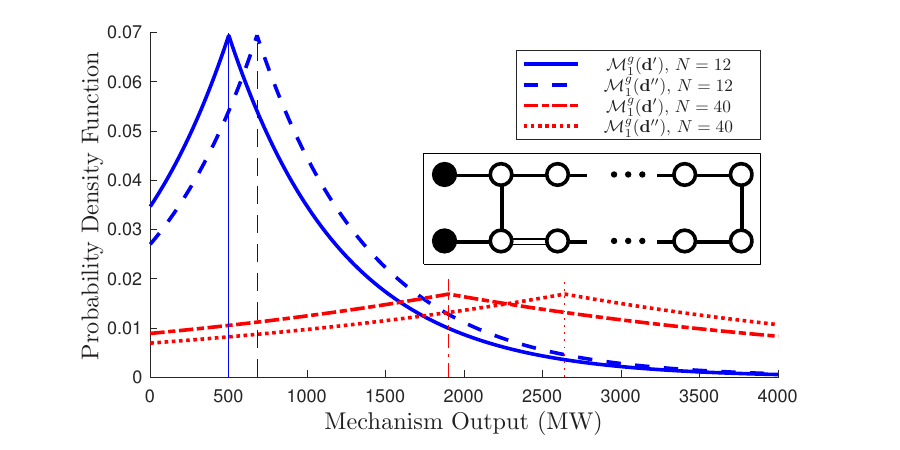}
\vspace{-2em}
\caption{Differential privacy for a ring network, shown in the embedded diagram. 
Black nodes represent generators and white nodes represent loads.
The figure shows the density functions of the aggregation mechanism for different network sizes.}
\vspace{-1em}
\label{fig:DP_cycle}
\end{figure}
\subsection{Bad Topology}
There are networks whose behavior
can be arbitrarily far from monotone, i.e., they are 
$(\delta, \varepsilon)$-monotone with large $\varepsilon$.
For these networks, differential privacy is only possible with the addition
of large noise, potentially rendering the output of the mechanism meaningless. 

One such network is shown in  Figure~\ref{fig:DP_cycle}. This network consists
of a cycle with $N$ buses, with generators on two adjacent buses
(black nodes). 
The branch indicated by the double-line edge is the only bottleneck
where the line flow constraint is binding.
It can be shown that this network is $(\Delta, (\nGL-4)\Delta)$-monotone for 
some positive $\Delta$.
This means that a change in load can be amplified $\nGL-4$ times in some generator, 
implying a large $L_1$-sensitivity. 
Figure~\ref{fig:DP_cycle} shows that to achieve $(20,0.5)$-differential privacy, a far 
bigger noise is required than in the monotone case. As $N$ increases, the 
density function becomes flatter.
When $N=40$ buses, the density function in Figure~\ref{fig:DP_cycle} 
is close to a uniform distribution, i.e., the mechanism hardly discloses any useful information.

\section{Conclusion}
We have proposed a differential privacy model for OPF data 
in power systems.  We have introduced the notion of  
monotonicity of the $\OPF$ operator and used it 
to determine the amount of noise needed to preserve
differential privacy for aggregation queries. 
We have also shown that, for the aggregation query, the level of differential privacy is independent of the number of aggregation regions and the number of buses in a region. 
We also derive the required noise level for arbitrary queries with bounded Jacobian values.
Future work will look at how these results can be applied to the design of new mechanisms.







\bibliographystyle{IEEEtran}
\bibliography{my-bibliography}

\appendix
 
\subsection{Validating Assumption \ref{A:vectorf}}\label{app:Assumption1}
Let $\taueq\in\Real^{\nGL+1}$ be the vector of Lagrangian multipliers 
associated with equality constraints \eqref{eq:opf1.b}, \eqref{eq:opf1.c},
and $(\lambdap,\lambdam)$ and $(\mup,\mum)$ be the Lagrangian multipliers 
associated with inequalities \eqref{eq:opf1.d} and \eqref{eq:opf1.e} respectively. 
As \eqref{eq:opf1} is a linear program \cite{Ber97introduction}, the following KKT condition
holds at an optimal point when \eqref{eq:opf1} is feasible.
\begin{subequations}
\begin{eqnarray}
&& \eqref{eq:opf1.b}-\eqref{eq:opf1.e}
\label{eq:KKT.a}\\
&& \bf{0}=M^{\T}\taueq+\CM\B(\mup-\mum)
\label{eq:KKT.b}\\
&&  -\f=-[\taueq_1,\taueq_2,\cdots,\taueq_{\nG}]^{\T}+\lambdap-\lambdam
\label{eq:KKT.c}\\
&& \mup,\mum,\lambdap,\lambdam\geq 0
\label{eq:KKT.d}\\
&& \mup^{\T}(\B\CM^{\T}\ang-\overline{\pflow})=\mum^{\T}(\underline{\pflow}-\B\CM^{\T}\ang)=0
\label{eq:KKT.e}\\
&& \lambdap^{\T}(\sgen-\genulim)=\lambdam^{\T}(\genllim-\sgen)=0,
\label{eq:KKT.f}
\end{eqnarray}
\label{eq:KKT}
\end{subequations}
where
\begin{align*}
\M:=\left[
\begin{array}{c}
\CM\B\CM^{\T}\\
1~0~0~\cdots~0
\end{array} 
\right]
\end{align*}
is an $(\nGL+1)$-by-$\nGL$ matrix with rank $\nGL$. 
Condition \eqref{eq:KKT.a} corresponds to primal feasibility, condition \eqref{eq:KKT.d} corresponds to dual feasibility, conditions \eqref{eq:KKT.e}, \eqref{eq:KKT.f} correpond to complementary slackness,
and conditions \eqref{eq:KKT.b}, \eqref{eq:KKT.c} correspond to stationarity \cite{Boy04convex}.

The following Proposition shows that for a fixed network, it is easy to find an objective vector $\f$ such that \eqref{eq:opf1} will always have a unique solution for a reasonable $\sload$.
When Assumption~\ref{A:vectorf} does not hold, Proposition~\ref{fisdense} implies that we can always perturb $\f$ a little such that the assumption is valid.
\begin{proposition}\label{fisdense}
$\setf$ is dense in $\Real^{\nG}$.
\end{proposition}

\begin{proof} 
We first show that for a fixed network $(\B, \CM, \para)$ and $\sload\in\setsl$, if the primal optimal solution to $\eqref{eq:opf1}$ is not unique for some $\f$, then there must exist $\taueq,\mup,\mum,\lambdap,\lambdam$ such that \eqref{eq:KKT} holds but \eqref{eq:multiplier1} does not.
Geometrically, an LP has multiple optimal solutions if and only if 
the objective vector is normal to some hyperplane defined by equality constraints 
and the set of binding inequality constraints. 
In our case, the objective vector $[\f^{\T}, \bf{0}^{\T}]^{\T}$ is a $\nG+\nGL$ dimensional vector.
As there are $\nGL+1$ linearly independent equality constraints in 
\eqref{eq:opf1.b}, \eqref{eq:opf1.c},\footnote{Here, independence means that the gradient of the equalities \eqref{eq:opf1.b} and \eqref{eq:opf1.c} with respect to $[(\sgen)^{\T},\ang^{\T}]^{\T}$ has full column rank $\nGL+1$.}
it is therefore enough to take $\leq\nG-2$ inequality constraint vectors, along with the above $\nGL+1$ to represent $[\bf{0}^{\T},\f^{\T}]^{\T}$. 
That is to say, there exist Lagrange multipiers that satisfy \eqref{eq:KKT}, and there are at most $\nG-2$ non-zero coefficients in $\mup,\mum,\lambdap,\lambdam$, i.e.,
\begin{align}\label{eq:multiplier2}
\|\mup\|_0+\|\mum\|_0+\|\lambdap\|_0+\|\lambdam\|_0< \nG-1.
\end{align}
Thereby, \eqref{eq:multiplier1} implies uniqueness, and we have
\begin{align}\label{eq:omegaf}
\setf=\{\f~|~\forall \sload \in\setsl, \text{the solutions of}~\eqref{eq:KKT}~\text{satisfy}~\eqref{eq:multiplier1}\}.
\end{align}
For $\setS\subseteq[\nE]$, $\setT\subseteq[\nG]$ such that $|\setS|+|\setT|\leq\nG-2$, we construct $\setQ(\setS,\setT)$ to be the set of $\f$ such that 
$\exists \taueq\in\Real^{\nGL+1},\mupm\in\Real^{\nE},\lambdapm\in\Real^{\nG}$ satisfying:
\begin{subequations}
\begin{eqnarray}
&& \bf{0}=M^{\T}\taueq+\CM\B\mupm
\label{eq:setQ.a}\\
&&  -\f=-[\taueq_1,\taueq_2,\cdots,\taueq_{\nG}]^{\T}+\lambdapm
\label{eq:setQ.b}\\
&&  \mupm_i\neq 0 \Rightarrow i\in\setS
\label{eq:setQ.c}\\
&&  \lambdapm_i\neq 0 \Rightarrow i\in\setT.
\label{eq:setQ.d}
\end{eqnarray}
\label{eq:setQ}
\end{subequations} 
When $\setS$ and $\setT$ are fixed, the vector $\CM\B\mupm$ takes value in an $|\setS|$ dimensional subspace. Since $\rank(\M)=\nGL$, the possible values of $\taueq$ must fall within an $|\setS|+1$ dimensional subspace.
Therefore, \eqref{eq:setQ.b} implies that $\f$ must be in an $|\setS|+1+|\setT|\leq\nG-1$ dimensional subspace, and hence $\interior(\close(\setQ(\setS,\setT))=\emptyset$. The set
\begin{align}
\bigcup\limits_{\tiny
\substack{\setS\subseteq[\nE],\setT\subseteq[\nG]\\
|\setS|+|\setT|\leq\nG-2 }}
\setQ(\setS,\setT)
\end{align}
is the union of finitely many nowhere dense sets and thereby is nowhere dense itself in $\Real^{\nG}$. On the other hand, \eqref{eq:omegaf} and \eqref{eq:setQ} imply that
\begin{align}
\Bigg(
\bigcup\limits_{\tiny
\substack{\setS\subseteq[\nE],\setT\subseteq[\nG]\\
|\setS|+|\setT|\leq\nG-2 }}
\setQ(\setS,\setT)
\Bigg)^{\co}\subseteq\setf
\end{align}
and hence $\setf$ is dense in $\Real^{\nG}$.

\subsection{Validating Assumption \ref{A:derivative}}\label{app:Assumption_der}
We first have the following proposition.

\begin{proposition}\label{limitsdense}
Let $\setparar\subseteq\setpara$ be the set such that $\forall \para\in\setparar$, the set
$\setslr(\para)$ 
is dense in $\setsl(\para)$. Then $\setparar$ is dense in $\setpara$.
\end{proposition}

The proof is given at the end of this subsection.

We will then validate that the following assumption is valid.
\begin{assumption}\label{A:limits}
The parameter $\para$ for the limits of generations and branch power flows is assumed to be in $\setparar$, as defined in Proposition \ref{limitsdense}. 
\end{assumption}

If Assumption~\ref{A:limits} does not hold, Proposition~\ref{limitsdense} implies that we can always perturb $\para$ such that the assumption holds.
When Assumption \ref{A:vectorf} and Assumption \ref{A:limits} hold, then Assumption \ref{A:derivative} is directly implied by the results given in \cite{Fia1976sensitivity,Fia1983introduction}. 
Next, we prove Proposition \ref{limitsdense}.

\begin{proof}
Consider the power equations below:
\begin{align}\label{eq:constraints}
\mathbf{T}\ang:=
\left[
\begin{array}{c}
\CM\B\CM^{\T}\\
\B\CM^{\T}
\end{array} 
\right]\cdot\ang=
\left[
\begin{array}{c}
\sgen\\
-\sload\\
\pf
\end{array} 
\right].
\end{align}
Proposition \ref{fisdense} and Assumption \ref{A:vectorf} show that there will always be at least $\nG-1$ binding inequality constraints as each non-zero multiplier will force one inequality constraint to be binding.
A constraint is binding means some $\sgen_i$ equals either $\genulim_i$ or $\genllim_i$ (as in the upper $\nG$ rows in \eqref{eq:constraints}),
or some $\pf_i$ equals either $\overline{\pflow}_i$ or $\underline{\pflow}_i$ (as in the lower $\nE$ rows in \eqref{eq:constraints}).
We have $\rank(\mathbf{T})=\nGL-1$.
We will first use the following procedure to construct a new set $\setparar'$.
\begin{enumerate}[I.]
\item $\setparar' \leftarrow \setpara$
\item For each $\setS\subseteq[\nG]\cup[\nGL+1,\nGL+\nE]$, construct $\mathbf{T}_{\setS}$.
\begin{enumerate}[a)]
\item If $\rank(\mathbf{T}_{\setS})=|\setS|$, then continue to the next $\setS$.
\item If $\rank(\mathbf{T}_{\setS})<|\setS|$, then consider
\begin{align}\label{eq:Gamma}
\!\!\Gamma:=\!\prod\limits_{\mathclap{\tiny \substack{i\in\setS\cap[\nG]}}}
\{\base_{\scriptscriptstyle i},\base_{\scriptscriptstyle\nG+i}\}\!\times
\!\prod\limits_{\mathclap{\tiny
\substack{j\in[E]\\j+\nGL\in\setS}}}
\{\base_{\scriptscriptstyle 2\nG+j},\base_{\scriptscriptstyle 2\nG+\nE+j}\}
\end{align}
and update $\setparar'$ as
\begin{align}\label{eq:setupdate}
\setparar' \leftarrow\setparar' \setminus
\bigcup\limits_{\gamma\in\Gamma}\big\{\para|\exists\ang,\st~\gamma^{\T}\para=\mathbf{T}_{\setS}\ang\big\}.
\end{align}
\end{enumerate}
\item Return $\setparar'$.
\end{enumerate}

In the above procedure, an $n$-tuple of vectors is also regarded as a matrix of $n$ columns.\footnote{Hence, each $\gamma\in\Gamma$ can also be regarded as a $(2\nG+2\nE)$-by-$|\setS|$ matrix.}
Since $\gamma\in\Gamma$ is of rank $|\setS|$ and $\{\mathbf{T}_{\setS}\ang|\forall\ang\in\Real^{\nGL}\}$ defines a subspace with $\leq|\setS|-1$ dimensions,
each set of $\{\para|\exists\ang,\st~\gamma^{\T}\para=\mathbf{T}_{\setS}\ang\}$ in \eqref{eq:setupdate} is a subspace with dimension strictly lower than $2\nG+2\nE$.
Using the same technique as in the proof of Proposition \ref{fisdense}, we have that $\setparar'$ is dense in $\setpara$.
It is sufficient to show that $\setparar'\subseteq\setparar$.

In fact, $\forall\para\in\setparar'$, if for some $\sload\in\setsl(\para)$, the optimal solution to \eqref{eq:opf1} has $\geq \nG$ binding inequality constraints, 
then we use $\setS\subseteq[\nG]\cup[\nGL+1,\nGL+\nE],|\setS|=\nG$ again to denote the indices of any $\nG$ binding inequality constraints.
As those $\nG$ inequality constraints are binding, there must exist $\ang\in\Real^{\nGL}$ and $\gamma\in\Gamma$, as defined in \eqref{eq:Gamma}, such that $\gamma^{\T}\para=\mathbf{T}_{\setS}\ang$.
According to \eqref{eq:setupdate}, $\rank(\mathbf{T}_{\setS})$ must be exactly $\nG$. Plugging the optimal $\ang$, as well as the binding limits indexed by some $\gamma\in\Gamma$, into \eqref{eq:constraints}, we have
\begin{subequations}
\begin{eqnarray}
&& \gamma^\T\para=\mathbf{T}_{\setS}\ang\\
&& -\sload=\mathbf{T}_{[\nG+1,\nGL]}\ang.
\end{eqnarray}
\label{eq:slplane}
\end{subequations}
For each $\gamma\in\Gamma$, as $\rank(\mathbf{T}_{\setS})=\nG$ but $\rank(\mathbf{T})=\nGL-1$, the set $\{\sload~|~\exists\ang,\eqref{eq:slplane}~\text{holds}\}$ has less dimension than $\Real^{\nL}$ and is thereby nowhere dense in $\setsl$.
\footnote{It is due to $\clos(\inte(\setsl))=\clos(\setsl)$. The detailed proof of this equality is omitted due to space limit}
As a result,
\begin{align*}
\setslr\supseteq\setsl\setminus\bigcup\limits_{\gamma\in\Gamma}
\{\sload~|~\exists\ang,\eqref{eq:slplane}~\text{holds for $\gamma$}\}
\end{align*}
must be dense in $\setsl$. Therefore, $\setparar'\subseteq\setparar$ and $\setparar$ is dense in $\setpara$.
\end{proof}

Finally, we have two corollaries of Proposition \ref{limitsdense}.

\begin{corollary}
In Proposition \ref{limitsdense}, $\setsl\setminus\setslr$ can be covered by the union of finitely many subspaces.
\end{corollary}
\begin{corollary}\label{Co:independent}
For any $\sload\in\setslr$, the $\nG-1$ binding inequalities in \eqref{eq:opf1}, along with $\nGL+1$ equality constraints, are independent.
\end{corollary}

 \subsection{Proof of Theorem \ref{Thm:treeTopl}}\label{app:proof_of_Thm_4}
For fixed $u\in[\nL]$, Assumption \ref{A:derivative} shows that there exists $\kappa_u>0$ such that $\forall\omega_u\in(-\kappa_u,\kappa_u)$, $\hatsload:=\sload+\omega_u\base_u$ satisfies 
\begin{align}
\label{eq:setunchanged}
\setSgen(\hatsload)=\setSgen(\sload),~
\setSbra(\hatsload)=\setSbra(\sload).
\end{align}
It is sufficient to show $\OPF(\hatsload)\geq\OPF(\sload)$ when $\omega_u\in(0,\kappa_u)$ for any fixed $u$.\footnote{By symmetry, we will have $\OPF(\hatsload)\leq\OPF(\sload)$ when $\omega_u\in(-\kappa_u,0)$ for any fixed $u$.}

Since $\setSbra\subseteq\EdgeI$, Proposition \ref{limitsdense} implies that $|\setSbra\cap\EdgeI|=|\setSbra|$.
Thereby $\setSbra$ splits $\Graph$ into $|\setSbra|+1$ connected components $\Graph_1,\dots,\Graph_{|\setSbra|+1}$, and each component has vertices $\Vertex_i$ and edges $\Edge_i$.

We first show that 
\begin{align}\label{eq:genDist}
\forall i\in[|\setSbra|+1],~|\Vertex_i\cap([\nG]\backslash\setSgen)|=1.
\end{align}
Since $\cup_{i=1}^{|\setSbra|+1}\Vertex_i=\Vertex\supseteq [\nG]\backslash\setSgen$, and 
\begin{align}
\nonumber
&|[\nG]\backslash\setSgen|=\nG-|\setSgen|\\
\nonumber
=&\nG-(\nG-1-|\setSbra|)=|\setSbra|+1,
\end{align}
if \eqref{eq:genDist} does not hold, then there must exist $i\in[\nG]$ such that $\Vertex_i\cap([\nG]\backslash\setSgen)=\emptyset$ and thus $\Vertex_i\cap[\nG]\subseteq\setSgen$.
Now, for component $\Graph_i$, power flow equations imply that
\begin{align}\label{eq:compDepend}
\sum\limits_{\scriptscriptstyle j\in\Vertex_i\cap\VertexG}\sgen_j-\sum\limits_{\scriptscriptstyle j\in\Vertex_i\cap\VertexL}\sload_{\scriptscriptstyle j-\nG}=
\sum\limits_{\tiny \substack{e:e\in\Edge,\\ \sum\limits_{\mathclap{k\in\Vertex_i}}\CM_{k,e}=1}}\pf_e-
\sum\limits_{\tiny \substack{e:e\in\Edge,\\ \sum\limits_{\mathclap{k\in\Vertex_i}}\CM_{k,e}=-1}}\pf_e.
\end{align}
In \eqref{eq:compDepend}, for $j\in\Vertex_i\cap\VertexG$, we have $\sgen_j\in\{\genllim,\genulim\}$ as $\Vertex_i\cap[\nG]\subseteq\setSgen$.
On the other hand, for $e\in\Edge$ such that $\sum_{k\in\Vertex_i}\CM_{k,e}=\pm 1$, $e$ must be the bridge connecting $\Graph_i$ and some other component, and thereby $e$ is in the cut $\setSbra$.
By definition, we have $\pf_e\in\{\underline{\pflow},\overline{\pflow}\}$. 
Since all the generators and branch power flows involved in \eqref{eq:compDepend} are binding, it contradicts to Corollary \ref{Co:independent} and therefore \eqref{eq:genDist} always holds.

Now let $\loosegen$ be the mapping such that $\loosegen(\Graph_i)$ is the unique generator in $\Vertex_i\cap([\nG]\backslash\setSgen)$ for each $i\in[|\setSbra|+1]$.
For any fixed $u\in[\nL]$ and $\omega_u\in(0,\kappa_u)$, we will prove that $\OPF_v(\hatsload)\geq\OPF_v(\sload)$ for each $v\in[\nG]$ by discussing the following three possible situations that may arise.
Assume $u+\nG\in\Vertex_k$ for $k\in[|\setSbra|+1]$.
\begin{itemize}
\item If $v\in\setSgen(\sload)$, then \eqref{eq:setunchanged} implies $v\in\setSgen(\hatsload)$ as well. 
Since $\OPF$ is continuous over $\omega_u\in(-\kappa_u,\kappa_u)$ and $\genllim<\genulim$, there must be $\OPF_v(\hatsload)=\OPF_v(\sload)$.
\item If $v=\loosegen(\Graph_k)$, then similar to \eqref{eq:compDepend} we have
\begin{align}\label{eq:compare}
\nonumber
&\sum\limits_{\scriptscriptstyle j\in\Vertex_k\cap\setSgen}\OPF_j(\sload)+\OPF_v(\sload)-\sum\limits_{\scriptscriptstyle j\in\Vertex_k\cap\VertexL}\sload_{\scriptscriptstyle j-\nG}\\
\nonumber
=&\sum\limits_{\tiny \substack{e:e\in\Edge,\\ \sum\limits_{\mathclap{l\in\Vertex_k}}\CM_{l,e}=1}}\pf_e-
\sum\limits_{\tiny \substack{e:e\in\Edge,\\ \sum\limits_{\mathclap{l\in\Vertex_k}}\CM_{l,e}=-1}}\pf_e\\
=&\sum\limits_{\scriptscriptstyle j\in\Vertex_k\cap\setSgen}\OPF_j(\hatsload)+\OPF_v(\hatsload)-\sum\limits_{\scriptscriptstyle j\in\Vertex_k\cap\VertexL}\hatsload_{\scriptscriptstyle j-\nG}
\end{align}
As $\sload$ and $\hatsload$ only differ at load $u$ and $\OPF_j(\hatsload)=\OPF_j(\sload)$ for all $j\in\Vertex_k\cap\setSgen$ as shown above, \eqref{eq:compare} can be simplified to
\begin{align}
\nonumber
\OPF_v(\hatsload)-\OPF_v(\sload)=\hatsload_u-\sload_u=\omega_u>0,
\end{align}
and therefore $\OPF_v(\hatsload)>\OPF_v(\sload)$.
\item If $v=\loosegen(\Graph_{k'})$ for some $k'\neq k$, then \eqref{eq:compare} still holds for $\Graph_{k'}$ but $\sload$ and $\hatsload$ are identical for loads in $\Graph_{k'}$.
Hence we have $\OPF_v(\hatsload)=\OPF_v(\sload)$.
\end{itemize}
Putting this together, we conclude that $\OPF_v(\hatsload)\geq\OPF_v(\sload)$ for all $v\in[\nG]$.
\end{proof}

\end{document}